\documentclass[preprint,5p,times,authoryear]{elsarticle}
\usepackage{natbib}

\usepackage{amsmath} 
\usepackage{amsfonts} 
\usepackage{bm} 
\usepackage{calc} 
\usepackage{float}
\usepackage{subcaption} 
\usepackage{hyperref} 

\usepackage{svg}
\usepackage{siunitx}
\usepackage{amssymb,moreverb,pgfplots,tikz}
\pgfplotsset{compat=1.12}
\usepackage{rotating}
\usepackage{enumitem}
\usepackage{cleveref}

\usepackage{algorithm}
\usepackage{algorithmicx}
\usepackage{algpseudocode}

\makeatletter
\renewcommand*\env@matrix[1][*\c@MaxMatrixCols c]{%
  \hskip -\arraycolsep
  \let\@ifnextchar\new@ifnextchar
  \array{#1}}

\newcommand{\norm}[1]{\left\lVert#1\right\rVert}

\definecolor{mycolor1}{rgb}{0.00000,0.44700,0.74100}%
\definecolor{mycolor2}{rgb}{0.92900,0.69400,0.12500}%
\definecolor{MatlabBlue}{rgb}    {0     , 0.4470, 0.7410}
\definecolor{MatlabRed}{rgb}     {0.8500, 0.3250, 0.0980}
\definecolor{MatlabYellow}{rgb}  {0.9290, 0.6940, 0.1250}
\definecolor{MatlabPurple}{rgb}  {0.4940, 0.1840, 0.5560}
\definecolor{MatlabGreen}{rgb}   {0.4660, 0.6740, 0.1880}
\definecolor{MatlabBabyBlue}{rgb}{0.3010, 0.7450, 0.9330}
\definecolor{MatlabGray}{rgb}{0.5, 0.5, 0.5}
\definecolor{MatlabLightGray}{rgb}{0.75, 0.75, 0.75}
\definecolor{MatlabBlack}{rgb}{0, 0, 0}

\definecolor{MatlabLightGray4}{rgb}{0.875, 0.875, 0.875}
\definecolor{MatlabLightGray3}{rgb}{0.85, 0.85, 0.85}
\definecolor{MatlabLightGray2}{rgb}{0.775, 0.775, 0.775}
\definecolor{MatlabLightGray1}{rgb}{0.7, 0.7, 0.7}
\definecolor{MatlabGray30}{rgb}{0.3, 0.3, 0.3}
\definecolor{MatlabGray40}{rgb}{0.4, 0.4, 0.4}
\definecolor{MatlabGray50}{rgb}{0.5, 0.5, 0.5}
\definecolor{MatlabGray60}{rgb}{0.6, 0.6, 0.6}
\definecolor{MatlabGray70}{rgb}{0.7, 0.7, 0.7}
\definecolor{MatlabGray80}{rgb}{0.8, 0.8, 0.8}
\definecolor{MatlabGray90}{rgb}{0.9, 0.9, 0.9}

\definecolor{Red}{rgb}{1 0 0}
\definecolor{Black}{rgb}{0 0 0}

\definecolor{myblue}{rgb}{0 0 0}

\newcommand{\tikzline}[1]{(\protect\tikz[baseline=-0.6ex,x=1pt,y=1pt]{ \protect\draw[#1,thick] [-] (0,0) -- (10,0);})}
\newcommand{\tikzdashedline}[1]{(\protect\tikz[baseline=-0.6ex,x=0.9pt,y=1pt]{ \protect\draw[#1,thick,dashed] [-] (0,0) -- (10,0);})}

\DeclareRobustCommand\encircle[1]{%
\tikz[baseline=(X.base)] 
   \node (X) [draw, shape=circle, inner sep=0] {\strut #1};}

\hypersetup{pdfauthor={Name}} 

\usepackage{amsthm}

\newtheorem{theorem}{Theorem}

\newtheorem{remark}{Remark}[section]

\usepackage{graphicx}

\usepackage{fancyhdr}
\begin{document}
   
\begin{titlepage}
    \title{Nullspace-based Fault Diagnosis for Closed-Loop Mechatronic Systems\\ with Application to Semiconductor Equipment}
    
    \author[1]{Koen Classens\corref{cor1}}
    \ead{k.h.j.classens@gmail.com}
    \author[2]{Jeroen van de Wijdeven}
    \author[1]{W. P. M. H. (Maurice) Heemels}
    \author[1,3]{Tom Oomen}
    
    \cortext[cor1]{Corresponding author}
   
    \address[1]{Eindhoven University of Technology, Department of Mechanical Engineering, Control Systems Technology, Eindhoven, The Netherlands}
    \address[2]{ASML, Veldhoven, The Netherlands}
    \address[3]{Delft University of Technology, Department of Mechanical Engineering, Delft Center for Systems and Control, Delft, The Netherlands}

\begin{abstract}
Fault detection and isolation (FDI) systems are critical for modern mechatronic production equipment, as their continuous operation is heavily dependent on the ability to detect and isolate faults in a timely and efficient manner. The aim of this paper is to address closed-loop aspects for linear systems and enable the application of well-known nullspace-based FDI synthesis conditions to mechatronic systems subject to actuator and sensor faults. These tailored FDI synthesis conditions are applied to a large-scale prototype wafer stage, showcasing the proposed approach through real experiments, thereby underlining the usefulness of the derived synthesis conditions for a wide range of production machines and scientific instruments.
\end{abstract}


\end{titlepage}
    
\maketitle
\thispagestyle{fancy}

\section{Introduction}
The high-tech industry is moving towards a new approach to manage maintenance of its mechatronic production equipment, with a strong focus on predictive strategies to reduce the high expenses related to unplanned downtime. In this transition, real-time diagnosis of closed-loop systems is expected to play a crucial role to facilitate effective and targeted maintenance. Historically, these fault diagnosis systems are well established in safety-critical domains such as aerospace and automotive, whereas this article explores their application in linear multi-input multi-output (MIMO) mechatronic systems.

Driven by the stringent requirements of safety-critical systems, significant progress has been made over the past decades in the areas of fault diagnosis and fault tolerance. This progress is documented in numerous surveys such as \cite{isermannModelbasedFaultdetectionDiagnosis2005,hwangSurveyFaultDetection2010,gaoSurveyFaultDiagnosis2015a,Li2020}. Representative developments include observer-based approaches \citep{frank1997survey,zhangIntegratedTradeoffDesign2008,wangLMIApproachIndex2007}, factorization-based methods \citep{Ding2000}, parameter estimation techniques \citep{gertlerDiagnosingParametricFaults1995,classens2022fault}, and statistical approaches \citep{yinReviewBasicDataDriven2014}.

Recent research in fault diagnosis has increasingly focused on data-driven, machine learning, and hybrid approaches, driven by the growing availability of operational data and advances in artificial intelligence. Data-driven methods have shown strong capabilities in machinery health monitoring and prognostics \citep{Lei2018,Neupane2025}, while machine learning techniques, including deep learning and ensemble learning methods, have demonstrated promising performance for fault detection and classification in complex systems \citep{leiApplicationsMachineLearning2020,Mian2024,Leite2025}. In parallel, hybrid approaches that combine physics-based models with data-driven or knowledge-based techniques have gained significant attention \citep{Wilhelm2021,Rezamand2020}.

Due to the high financial stakes associated with downtime of mechatronic production equipment, transparency and deep understanding are required to support decision-making. Model-based approaches offer this and provide greater transparency and interpretability than black-box methods. In addition, the reproducible behavior and availability of accurately identified models, often developed during the design and integration phases, make model-based approaches particularly suitable for this application domain.

Fault detection and isolation (FDI) for closed-loop mechatronics poses several unique challenges. Generically applicable methods such as the nullspace-based approach \citep{vargaNewComputationalParadigms2013,vargaSolvingFaultDiagnosis2017} have been developed using an open-loop framework and have been successfully applied in closed-loop scenarios. However, in areas such as system identification, it is widely recognized that closed-loop behavior is critical and must be considered carefully \citep{vandenhofClosedloopIssuesSystem1998}.

Although important progress has been made in the field of FDI, the implications of the closed loop are not yet fully understood. While generic existence conditions for FDI filters are well established, further analysis is required to uncover fundamental limitations and opportunities. Moreover, experimental evidence supporting the applicability of such methods to large-scale industrial systems remains limited. The framework presented in this paper builds on the preliminary version reported in \cite[Chpt. 6]{classensFaultDiagnosisUncertain2024} and addresses closed-loop aspects relevant to mechatronic systems. The main contributions of this paper are the following.

\begin{enumerate}[label=C\arabic*]  \setlength{\itemsep}{0pt}
    \item Conditions are established under which open-loop nullspace-based FDI synthesis remains valid for closed-loop systems, and scenarios are identified in which a closed-loop formulation provides additional design freedom.
    \item Generic solvability and isolability conditions for approximate fault detection and isolation are translated into tailored criteria for systems subjected to actuator and sensor faults.
    \item The nullspace-based approach is experimentally validated on multi-input multi-output prototype wafer stage, demonstrating its applicability to industrial mechatronic systems.
\end{enumerate}

The remainder of this article is organized as follows. The problem is formulated in \Cref{chpt6:sec:PROBLEM FORMULATION}. In particular, the approximate fault detection and isolation problem (AFDIP) for closed-loop and open-loop systems. Next, the closed-loop aspects are clarified in \Cref{chpt6:sec:CLAspects}. The solution to the AFDIP is described in \Cref{chpt6:sec:SOLUTION}. Subsequently, the solvability requirements are examined and transformed to reveal fundamental properties in \Cref{chpt6:sec:DESIGN FOR ACTUATOR AND SENSOR FAULTS}, followed by a concise discussion. Experimental results are presented in \Cref{chpt6:sec:EXPERIMENTAL PROOF OF PRINCIPLE}. Finally, the findings are summarized and conclusions drawn in \Cref{chpt6:sec:CONCLUSION}.

\thispagestyle{empty}
\section{Problem formulation} \label{chpt6:sec:PROBLEM FORMULATION}
Consider the closed-loop output for linear continuous-time multi-input multi-output (MIMO) systems affected by additive faults, described in the Laplace domain, by
\begin{equation}
\begin{split}
y &= \left( I + G_{u} C \right)^{-1} G_{u} C r + \left( I + G_{u} C \right)^{-1} G_{d} d \\ &\qquad+ \left( I + G_{u} C \right)^{-1} G_{w} w + \left( I + G_{u} C \right)^{-1} G_{f} f,
\end{split}
\label{chpt6:eq:y_CL}
\end{equation}
and consider the control input, described by
\begin{equation}
\begin{split}
u &= \left( I + C G_{u} \right)^{-1} C r -\left( I + C G_{u} \right)^{-1} C G_{d} d \\ &\qquad -\left( I + C G_{u} \right)^{-1} C G_{w} w -\left( I + C G_{u} \right)^{-1} C G_{f} f,
\end{split}
\label{chpt6:eq:u_CL}
\end{equation}
see \Cref{chpt6:fig:CL_System}. In time domain, the output $y$ takes values in $\mathbb{R}^{n_{y}}$, the control input $u$ takes values in $\mathbb{R}^{n_{u}}$, the disturbance $d$ takes values in $\mathbb{R}^{n_{d}}$, the noise $w$ takes values in $\mathbb{R}^{n_{w}}$, and the fault vector $f$ takes values in $\mathbb{R}^{n_{f}}$. The transfer function matrices (TFMs) $G_{u}$, $G_{d}$, $G_{w}$, $G_{f}$, and the feedback controller $C$ are of corresponding dimensions.

A residual generator, described by a proper and stable TFM $Q := \begin{bmatrix}
Q_{y} & Q_{u}
\end{bmatrix}$, augments the closed-loop controlled system, and processes the known control input $u$ and measurable output signal $y$, see Figure \ref{chpt6:fig:CL_System}. Hence, the residual $\varepsilon$, taking values in $\mathbb{R}^{n_{\varepsilon}}$, is described by
\begin{equation}
\varepsilon = \begin{bmatrix} Q_{y} & Q_{u} \end{bmatrix} \begin{bmatrix}
y \\ u
\end{bmatrix},
\label{chpt6:eq:varepsilon_CL}
\end{equation}
and is used for decision making on the presence or absence of faults. The filter $Q$, which is to be designed, must meet specific requirements to enable the detection and isolation of faults. These requirements are imposed on the internal representation, derived by substituting the closed-loop output and input, described by \eqref{chpt6:eq:y_CL} and \eqref{chpt6:eq:u_CL}, into \eqref{chpt6:eq:varepsilon_CL}, which gives
\begin{equation*}
\varepsilon = \begin{bmatrix}
Q_{y} & Q_{u}
\end{bmatrix} \begin{bmatrix}
G_{u} C S & S G_{d} & S G_{w} & S G_{f} \\
C S & -C S G_{d} & -C S G_{w} & -C S G_{f}
\end{bmatrix} \begin{bmatrix}
r \\ d \\ w \\ f
\end{bmatrix},
\end{equation*}
with sensitivity function $S = \left(I + G_{u} C \right)^{-1}$. Alternatively, the internal representation of the closed-loop formulation is denoted as
\begin{equation}
\varepsilon = 
R_{r}^{\mathrm{cl}} r + R_{d}^{\mathrm{cl}} d + R_{w}^{\mathrm{cl}} w + R_{f}^{\mathrm{cl}} f,
\label{chpt6:eq:internal}
\end{equation}
where 
\begin{equation}
\begin{split}
&\begin{bmatrix}[c|c|c|c]
R_{r}^{\mathrm{cl}} & R_{d}^{\mathrm{cl}} & R_{w}^{\mathrm{cl}} & R_{f}^{\mathrm{cl}}
\end{bmatrix} :=\\ &\qquad \begin{bmatrix}
Q_{y} & Q_{u}
\end{bmatrix} \begin{bmatrix}[c|c|c|c]
G_{u} C S & S G_{d} & S G_{w} & S G_{f} \\
C S & -C S G_{d} & -C S G_{w} & -C S G_{f}
\end{bmatrix}, \label{chpt6:eq:compare}
\end{split}
\end{equation}
and the column of the fault-to-residual TFM $R_{f}^{\mathrm{cl}}$, corresponding to fault $f_{j}$, is denoted by $R_{f_{j}}^{\mathrm{cl}}$, where $j = 1, \ldots, n_{f}$.


%
\begin{figure}[tb]
\centering
\includegraphics[]{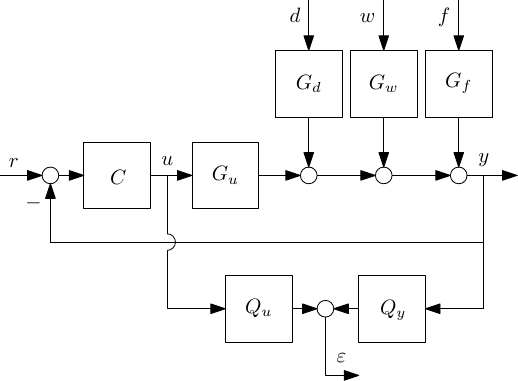}
\caption{Closed-loop controlled system equipped with a residual generator.}
\label{chpt6:fig:CL_System}
\end{figure}

In case the system is in open loop, see Figure \ref{chpt6:fig:OL_System}, the output is described by
\begin{equation}
y = G_{u} u + G_{d} d + G_{w} w + G_{f} f.
    \label{chpt6:eq:y_OL}
\end{equation}
In this scenario, the internal representation is described by
\begin{equation}
\varepsilon = 
R_{u}^{\mathrm{ol}} u + R_{d}^{\mathrm{ol}} d + R_{w}^{\mathrm{ol}} w + R_{f}^{\mathrm{ol}} f,
\label{chpt6:eq:OL_internal}
\end{equation}
where
\begin{equation}
\begin{split}
&\begin{bmatrix}[c|c|c|c]
R_{u}^{\mathrm{ol}} & R_{d}^{\mathrm{ol}} & R_{w}^{\mathrm{ol}} & R_{f}^{\mathrm{ol}}
\end{bmatrix} \\ &\qquad := \begin{bmatrix}
Q_{y} & Q_{u}
\end{bmatrix} \begin{bmatrix}[c|c|c|c]
G_{u} & G_{d} & G_{w} & G_{f} \\
I & 0 & 0 & 0
\end{bmatrix}, \label{chpt6:eq:OL_compare}
\end{split}
\end{equation}
and $R_{f_{j}}^{\mathrm{ol}}$ is the column of $R_{f}^{\mathrm{ol}}$ corresponding to the $j^{\mathrm{th}}$ fault.
\begin{figure}[tb]
\centering
\includegraphics[]{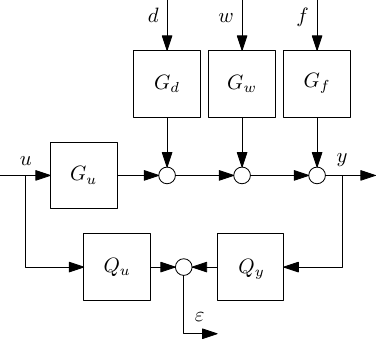}
\caption{Open-loop system equipped with a residual generator.}
\label{chpt6:fig:OL_System}
\end{figure}

\subsection{Approximate fault detection problem (AFDP)}
The approximate fault detection problem (AFDP) for closed-loop systems involves finding a filter $Q$ such that the internal representation \eqref{chpt6:eq:internal} is proper, stable, and satisfies the following requirements:
\begin{enumerate}[label=\alph*)] \setcounter{enumi}{0}
\item $R_{r}^{\mathrm{cl}} = 0$, \label{chpt6:cond:FD_CL_1}
\item $R_{d}^{\mathrm{cl}} = 0$, \label{chpt6:cond:FD_CL_2}
\item $R_{w_{j}}^{\mathrm{cl}} \approx 0, \:\: j = 1, \ldots, n_{w}, \:\: \mathrm{with} \: \max_{j} \left( \mathrm{Re} \left( \lambda \left( R_{{w}_{j}}^{\mathrm{cl}} \right) \right) \right) < 0$, \label{chpt6:cond:FD_CL_3}
\item $R_{f_{j}}^{\mathrm{cl}} \neq 0, \:\: j = 1, \ldots, n_{f}, \:\: \mathrm{with} \: \max_{j} \left( \mathrm{Re} \left( \lambda \left( R_{{f}_{j}}^{\mathrm{cl}} \right) \right) \right) < 0$, \label{chpt6:cond:FD_CL_4}
\end{enumerate}
where $\mathrm{Re} \left( \lambda \left( \cdot \right) \right)$ denotes the real part of the poles of the TFMs. The first requirements \ref{chpt6:cond:FD_CL_1} and \ref{chpt6:cond:FD_CL_2} enforce decoupling of the reference input $r$ and the disturbance $d$ from the residual $\varepsilon$, in the sense that $r$ and $d$ do not propagate to the residual signal. That is, when $R_{r}^{\mathrm{cl}} = 0$ and $R_{d}^{\mathrm{cl}} = 0$, the residual $\varepsilon$ is completely insensitive to the reference and external disturbances, and any nonzero residual can therefore not be attributed to these signals. The requirements \ref{chpt6:cond:FD_CL_3} and \ref{chpt6:cond:FD_CL_4} are formalized as part of an optimization problem in \Cref{chpt6:sec:SOLUTION}. The solution to this optimization problem simultaneously minimizes the transfer of measurement noise $w$ to the residual $\varepsilon$, while ensuring that the transfer from faults $f$ to $\varepsilon$ remains nonzero. This guarantees that the residual is predominantly driven by fault effects, while being robust to noise.
Note that if the requirements \ref{chpt6:cond:FD_CL_1}, \ref{chpt6:cond:FD_CL_2}, and \ref{chpt6:cond:FD_CL_4} hold, the requirement \ref{chpt6:cond:FD_CL_3} can always be achieved by scaling the filter $Q$.

The disturbance inputs targeted for exact decoupling, and for which such decoupling is presumably achievable, are represented by $G_{d}$. In contrast, noise and disturbances not covered by these inputs are represented by $G_{w}$. In practice, the distinction between $G_{d}$ and $G_{w}$ lies solely in how they are handled in the solution of the AFDP.

When the system is considered in open loop, see Figure \ref{chpt6:fig:OL_System}, the goal is to decouple the plant input $u$. To this end, requirement \ref{chpt6:cond:FD_OL_1} is replaced by $R_{u}^{\mathrm{ol}} = 0$. In this case, the AFDP involves finding a filter $Q$ such that the internal representation \eqref{chpt6:eq:OL_internal} is proper, stable, and satisfies the following requirements:
\begin{enumerate}[label=\alph*)] \setcounter{enumi}{0}
\item $R_{u}^{\mathrm{ol}} = 0$, \label{chpt6:cond:FD_OL_1}
\item $R_{d}^{\mathrm{ol}} = 0$, \label{chpt6:cond:FD_OL_2}
\item $R_{w_{j}}^{\mathrm{ol}} \approx 0, \:\: j = 1, \ldots, n_{w}, \:\: \mathrm{with} \: \max_{j} \left( \mathrm{Re} \left( \lambda \left( R_{{w}_{j}}^{\mathrm{ol}} \right) \right) \right) < 0$, \label{chpt6:cond:FD_OL_3}
\item $R_{f_{j}}^{\mathrm{ol}} \neq 0, \:\: j = 1, \ldots, n_{f}, \:\: \mathrm{with} \: \max_{j} \left( \mathrm{Re} \left( \lambda \left( R_{{f}_{j}}^{\mathrm{ol}} \right) \right) \right) < 0$. \label{chpt6:cond:FD_OL_4}
\end{enumerate}
\thispagestyle{empty}

\subsection{Approximate fault detection and isolation problem (AFDIP)}
To effectively isolate faults, a dedicated structure is enforced in the residual signals, with the aim of ensuring that each residual $\varepsilon_{i}$ is influenced by distinct set of faults. To this end, consider a $n_{\varepsilon}$-dimensional residual vector $\varepsilon$ and a structured bank of filters $Q^{(1)}$ to $Q^{(n_{\varepsilon})}$ as
\begin{equation}
Q = \begin{bmatrix}
Q^{(1)} \\ \vdots \\ Q^{(n_{\varepsilon})}
\end{bmatrix}. \label{chpt6:eq:bankQ}
\end{equation}
The corresponding fault-to-residual TFM in the closed-loop and open-loop scenario, $R_{f_{j}}^{\mathrm{cl}}$ and $R_{f_{j}}^{\mathrm{ol}}$ respectively, are structured accordingly as
\begin{equation}
R_{f} = \begin{bmatrix}
R_{f_{1}}^{(1)} & \ldots & R_{f_{n_{f}}}^{(1)} \\
\vdots & \ddots & \vdots \\
R_{f_{1}}^{(n_{\varepsilon})} & \ldots & R_{f_{n_{f}}}^{(n_{\varepsilon})}
\end{bmatrix}.
\end{equation}
Let $S_{R_{f}}$ be the corresponding $n_{\varepsilon} \times n_{f}$ matrix, which encodes the structure of $R_{f}$ as
\begin{subequations}
\begin{align}
S_{R_{f}} (i,j) = 1 \qquad & \mathrm{if}\qquad  R_{f_{j}}^{(i)} \neq 0 \\
S_{R_{f}} (i,j) = 0 \qquad & \mathrm{if}\qquad  R_{f_{j}}^{(i)} = 0
\end{align}
\end{subequations}

The approximate fault detection and isolation problem (AFDIP) requires determining a stable and proper filter $Q$ that, given a structure matrix $S$, additionally satisfies
\begin{enumerate}[label=\emph{\alph*})] \setcounter{enumi}{4}
\item $S_{R_{f}} = S$. \label{chpt6:cond:FD_OL_5}
\end{enumerate}


The selection of the structure matrix $S$ is an important factor in solving fault isolation problems. Notably, this selection is not unique; multiple configurations of $S$ can often yield equally satisfactory synthesis results. According to the nomenclature of \cite{gertlerFaultDetectionDiagnosis1998}, for a desired structure matrix $S$, the $i^{\mathrm{th}}$ row is referred to as the $i^{\mathrm{th}}$ specification, while the $j^{\mathrm{th}}$ column is known as the $j^{\mathrm{th}}$ fault signature.
The specifications describe the desired sensitivity pattern of each residual with respect to the set of faults and are used primarily during the synthesis procedure, i.e., the design of the filter bank $Q$. During synthesis, each filter $Q^{(i)}$ is constructed such that the corresponding residual $\varepsilon_i$ is sensitive only to the faults indicated by the nonzero entries in the $i^{\mathrm{th}}$ row of $S$, while being decoupled from the remaining faults.
In contrast, the fault signatures describe how each fault affects the residual vector and are used primarily during fault isolation. Specifically, the $j^{\mathrm{th}}$ fault signature is given by the $j^{\mathrm{th}}$ column of $S$ and represents the pattern of residuals that are expected to become nonzero when fault $f_j$ occurs. Fault isolation is achieved by comparing the observed pattern of fired (nonzero) residuals with the fault signatures encoded in the columns of $S$.
Although pairwise distinct fault signatures are sufficient to uniquely identify a single fault, a sparse structure matrix generally enables the isolation of a larger number of simultaneously occurring faults, as overlapping residual responses are reduced.

\section{Closed-loop aspects and solvability of the AFD(I)P}
\label{chpt6:sec:CLAspects}
\thispagestyle{empty}

This section has two objectives. First, it establishes the relationship between the open-loop and closed-loop formulations of the approximate fault detection and (isolation) problem (AFD(I)P), clarifying under which conditions the simpler open-loop formulation can be used for closed-loop systems. Second, it recalls standard solvability and isolability conditions, which are required later for residual generator synthesis. The first part contains new results, while the second part summarizes classical results, mainly following \cite{vargaSolvingFaultDiagnosis2017}, and is included for completeness and later use.

\subsection{Relationship between open-loop and closed-loop formulations}

A comparison of the open-loop and closed-loop decoupling requirements, in particular
\ref{chpt6:cond:FD_CL_1} and \ref{chpt6:cond:FD_CL_2}, reveals that certain residual generators $Q$ satisfy both formulations. In particular, residual generators obtained from the open-loop formulation may already fulfill the closed-loop requirements. The following two theorems formalize this relationship and constitute a contribution of this work.

\begin{theorem}
Consider $G_{u}$ with $n_{y} \geq n_{u}$, and let $G_{d}$, $G_{w}$ and $G_{f}$ be given transfer function matrices in a closed-loop system described by \eqref{chpt6:eq:y_CL} and \eqref{chpt6:eq:u_CL}, with a full normal rank controller $C$, and residual generator $Q$ in \eqref{chpt6:eq:varepsilon_CL}. Let $G_{u}$, $G_{d}$, $G_{w}$ and $G_{f}$ be the same transfer function matrices in the open-loop formulation described by \eqref{chpt6:eq:y_OL}. Then, any filter $Q$ that achieves $R_{u}^{\mathrm{ol}} = 0$ and $R_{d}^{\mathrm{ol}} = 0$ in the open-loop formulation also achieves $R_{r}^{\mathrm{cl}} = 0$ and $R_{d}^{\mathrm{cl}} = 0$ in the closed-loop formulation. Moreover, the corresponding nullspaces coincide, implying identical residual dynamics, i.e.,
\[
R_{w}^{\mathrm{ol}} = R_{w}^{\mathrm{cl}}, \quad
R_{f}^{\mathrm{ol}} = R_{f}^{\mathrm{cl}}.
\]
\label{chpt5:theoremOLisCL_nygeqnu}
\end{theorem}

\begin{proof}
See Appendix~\ref{app:proof:theoremOLisCL_nygeqnu}.
\end{proof}

Theorem~\ref{chpt5:theoremOLisCL_nygeqnu} shows that when $n_{y} \geq n_{u}$, the open-loop and closed-loop formulations are equivalent from the perspective of residual generation. Consequently, residual generators can be designed using the open-loop formulation even for closed-loop systems, yielding identical disturbance, noise, and fault responses while avoiding explicit dependence on the controller.

The next result addresses the complementary case $n_{y} < n_{u}$.

\begin{theorem}
\label{chpt5:theoremOLisCL_nystnu}
Consider $G_{u}$ with $n_{y} < n_{u}$, and let $G_{d}$, $G_{w}$ and $G_{f}$ be given transfer function matrices in a closed-loop system described by \eqref{chpt6:eq:y_CL} and \eqref{chpt6:eq:u_CL}, with a full normal rank controller $C$, and residual generator $Q$ in \eqref{chpt6:eq:varepsilon_CL}. Let $G_{u}$, $G_{d}$, $G_{w}$ and $G_{f}$ be the same transfer function matrices in the open-loop formulation described by \eqref{chpt6:eq:y_OL}. Then, the set of filters $Q$ satisfying $R_{u}^{\mathrm{ol}} = 0$ and $R_{d}^{\mathrm{ol}} = 0$ is a subset of the set of filters satisfying $R_{r}^{\mathrm{cl}} = 0$ and $R_{d}^{\mathrm{cl}} = 0$.
\end{theorem}

\begin{proof}
See Appendix~\ref{app:proof:theoremOLisCL_nystnu}.
\end{proof}

Theorem~\ref{chpt5:theoremOLisCL_nystnu} shows that for $n_{y} < n_{u}$ the closed-loop formulation admits additional design freedom due to a larger nullspace. While this additional freedom can be exploited by explicitly incorporating closed-loop operators, the open-loop formulation is adopted in the remainder of this paper for simplicity. The presented results extend naturally to the closed-loop case.

\subsection{Solvability and isolability conditions}

Consider the open-loop model \eqref{chpt6:eq:y_OL}. The approximate fault detection problem (AFDP) is solvable if and only if the system is completely fault detectable, i.e., each individual fault produces a nonzero residual that cannot be masked by disturbances. This property holds if and only if, for every fault channel $j = 1,\ldots,n_f$, the rank condition
\begin{equation}
\mathrm{rank} 
\begin{bmatrix}
G_{f_j} & G_d
\end{bmatrix}
>
\mathrm{rank}\, G_d
\label{chpt6:eq:CFD}
\end{equation}
is satisfied, where $G_{f_j}$ denotes the $j^{\mathrm{th}}$ column of $G_f$ and $\mathrm{rank}$ refers to the normal rank. The equivalence between \eqref{chpt6:eq:CFD} and solvability of the AFDP follows from \cite[Theorems 3.2, 3.7, 3.9]{vargaSolvingFaultDiagnosis2017} and will be used later when discussing existence of residual generators.

Beyond detectability, isolation requires that different faults produce distinguishable residual signatures. Let $S$ denote a binary $n_\varepsilon \times n_f$ structure matrix specifying which residual channels should respond to which faults. The system \eqref{chpt6:eq:y_CL} is $S$-fault isolable if and only if, for each residual channel $i = 1,\ldots,n_\varepsilon$, the rank condition
\begin{equation}
\mathrm{rank}
\begin{bmatrix}
G_d & \hat G_d^{(i)} & G_{f_j}
\end{bmatrix}
>
\mathrm{rank}
\begin{bmatrix}
G_d & \hat G_d^{(i)}
\end{bmatrix}, \: \forall j, S_{ij} \neq 0
\label{chpt6:eq:SFI}
\end{equation}
is satisfied, where $\hat G_d^{(i)}$ is formed from those columns of $G_f$ corresponding to $S_{ij}=0$. This condition ensures that fault $f_j$ cannot be hidden by disturbances and by faults that are required to be decoupled in residual channel $i$ \cite[Theorem 3.5]{vargaSolvingFaultDiagnosis2017}. Achievable structure matrices can be determined a priori using the procedure in \cite{vargaComputingAchievableFault2009}.

An important special case is \emph{strong fault isolability}, corresponding to $S = I_{n_f}$. In this case, all faults must be simultaneously isolable, which is equivalent to
\begin{equation}
\mathrm{rank}
\begin{bmatrix}
G_d & G_f
\end{bmatrix}
=
\mathrm{rank}\, G_d + n_f.
\label{chpt6:eq:strongIso}
\end{equation}
This condition is often not feasible in practice, especially when the number of sensors is limited.

A weaker requirement is \emph{weak fault isolability}, where faults are assumed to occur one at a time. The least restrictive structure that still allows isolation of all faults is the hollow structure $S = J_{n_f} - I_{n_f}$, where $J$ is the matrix of ones. In this case, the general $S$-fault isolability condition \eqref{chpt6:eq:SFI} reduces to the pairwise rank conditions
\begin{equation}
\mathrm{rank}
\begin{bmatrix}
G_d & G_{f_i} & G_{f_j}
\end{bmatrix}
>
\mathrm{rank}
\begin{bmatrix}
G_d & G_{f_i}
\end{bmatrix},
\quad i \neq j,
\label{chpt6:eq:SisolW}
\end{equation}
which ensure that any two distinct faults are distinguishable. Indeed, for the hollow structure one has $S_{ij}=0$ if $i=j$ and $S_{ij}=1$ otherwise, so that $\hat G_d^{(i)}$ in \eqref{chpt6:eq:SFI} reduces to $G_{f_i}$, yielding \eqref{chpt6:eq:SisolW}.

Finally, the approximate fault detection and isolation problem (AFDIP) with a prescribed structure matrix $S$ is solvable if and only if the system is $S$-fault isolable in the sense of \eqref{chpt6:eq:SFI} \cite[Theorems 3.10, 3.13]{vargaSolvingFaultDiagnosis2017}. These conditions will be used in the subsequent sections to characterize existence and structural limitations of the proposed residual generator designs.

\section{Solution to the AFDIP} \label{chpt6:sec:SOLUTION}
\thispagestyle{empty}

To synthesize a fault detection and isolation filter that solves the AFDIP, i.e., criteria \ref{chpt6:cond:FD_OL_1} to \ref{chpt6:cond:FD_OL_5} for an $S$-fault isolable system according to the rank condition \eqref{chpt6:eq:SFI}, the nullspace-based approach is followed \citep{vargaSolvingFaultDiagnosis2017}. Here, each filter $i=1,\ldots,q$ in the bank of filters \eqref{chpt6:eq:bankQ}, is factorized as
\begin{equation}
Q^{(i)} = Q_{3}^{(i)} Q_{2}^{(i)} Q_{1}^{(i)}.
\end{equation}
Each factor is interpreted as a partial synthesis result addressing specific requirements. First, for a specific $i$, define a new fault input $\hat{f}^{(i)}$ that contains the components $f_{j}$ for which $S_{ij} = 1$. Define $\hat{G}_{f}^{(i)}$ as the TFM formed by the columns $G_{f_{j}}$ for which $S_{ij} = 1$ and define $\hat{G}^{(i)}_{d}$ as the TFM formed by the columns $G_{f_{j}}$ for which $S_{ij} = 0$. Assume $n_{u} + n_{d} > 0$ and consider $Q^{(i)} = \bar{Q}_{1}^{(i)} Q_{1}^{(i)}$. Here, $\bar{Q}_{1}^{(i)} = {Q}_{3}^{(i)} {Q}_{2}^{(i)}$ and $Q_{1}^{(i)}$ is the first partial synthesis result which is a proper left rational nullspace basis satisfying
\begin{equation}
Q_{1}^{(i)} G^{(i)} = 0, \label{chpt6:eq:QiGi}
\end{equation} 
where
\begin{equation}
G^{(i)} := \begin{bmatrix}
G_{u} & G_{d} & \hat{G}^{(i)}_{d} \\ I_{n_{u}} & 0 & 0
\end{bmatrix}.
\end{equation}
The residual that remains can be written as
\begin{equation}
\varepsilon_{i} = \bar{Q}_{1}^{(i)} \bar{G}_{f}^{(i)} \hat{f}^{(i)} + \bar{Q}_{1}^{(i)} \bar{G}_{w}^{(i)} w,
\end{equation}
where $\bar{G}_{f}^{(i)}:= Q_{1}^{(i)} \begin{bmatrix}
\hat{G}_{f}^{(i)} \\ 0
\end{bmatrix}$ and  $\bar{G}_{w}^{(i)}:= Q_{1}^{(i)} \begin{bmatrix}
G_{w} \\ 0
\end{bmatrix}$. The factor $Q_{1}^{(i)}$ is computed using the approach in \cite{vargaComputingNullspaceBases2008}, ensuring that $\bar{Q}_{1}^{(i)} \bar{G}_{f}$ and $\bar{Q}_{1}^{(i)} \bar{G}_{w}$ are proper and stable. Computing $Q_{1}^{(i)}$ for all $i = 1,\ldots, q$ suffices to satisfy requirements a), b), and e). The remaining $\bar{Q}_{1}^{(i)}$ can be factored as $\bar{Q}_{1}^{(i)} = Q_{3}^{(i)} Q_{2}^{(i)}$, where $Q_{2}^{(i)}$ is a rational vector to construct a linear combination of the basis vectors of $Q_{1}^{(i)}$, which can be chosen such that $Q_{2}^{(i)} Q_{1}^{(i)}$ has the least possible McMillan degree \citep{Zhou1996a}. To this end, minimum dynamic cover algorithms are deployed \citep{kimuraGeometricStructureObservers1977,vargaReliableAlgorithmsComputing2003}. Alternatively, a different linear combination of basis vectors can be chosen via $Q_{2}^{(i)}$ at the cost of a higher McMillan degree.

Next, $Q_{3}^{(i)}$ is determined to maximize the fault-to-noise gap $\eta := \tfrac{\beta}{\gamma}$. An optimization-based approach is used to achieve the largest gap between fault detectability and noise attenuation. Let $\gamma > 0$ be an admissible level for the influence of $w$ on $\varepsilon_{i}$. An optimization problem to minimize the effect of $w$ and maximize the effect of $\hat{f}^{(i)}$ is posed as follows. Given $\gamma > 0$, determine $\beta > 0$ and a stable and proper fault detection filter $Q_{3}^{(i)}$ such that
\begin{equation}
\beta = \max_{Q_{3}^{(i)}} \left\{ \norm{ Q_{3}^{(i)} Q_{2}^{(i)} \bar{G}_{f} }_{\infty -} \Big| \norm{ Q_{3}^{(i)} Q_{2}^{(i)} \bar{G}_{w}}_{\infty} \leq \gamma \right\}, \label{chpt6:eq:optim_problem}
\end{equation}
where
\begin{equation}
\norm{ Q_{3}^{(i)} Q_{2}^{(i)} \bar{G}_{f} }_{\infty -} := \min_{1 \leq i \leq q } \norm{ Q_{3}^{(i)} Q_{2}^{(i)} \bar{G}_{f} }_{\infty}.
\end{equation}
This optimization problem is solved by determining $Q_{3}^{(i)}$ from a co-inner-outer factorization, $Q_{2}^{(i)} \bar{G}_{f} = G_{wo} G_{wi}$, where $G_{wo}$ is an invertible TFM which has only stable zeros and $G_{wi}$ is co-inner (i.e., $G_{wi} G_{wi}^{H} = I$ where $H$ denotes the Hermitian transpose). Setting $Q_{3}^{(i)} = \gamma G_{wo}^{-1}$ yields the optimal solution to \eqref{chpt6:eq:optim_problem}. Computing this outer factor generally involves solving a single Riccati equation \citep{liuOptimalSolutionsMultiobjective2007,Glover2011}.

\begin{remark}
The conditions of the AFDIP can be relaxed by replacing requirement e) such that the TFMs $R_{f_{j}}^{(i)}$ corresponding to $S_{ij} = 0$ are not required to be exactly zero, i.e., $R_{f_{j}}^{(i)} = 0$, but should be small $R_{f_{j}}^{(i)} \approx 0$. To this end, $\hat{G}^{(i)}_{d}$ is considered as part of $G_{w}$ which includes its contribution in the optimization problem \eqref{chpt6:eq:optim_problem}, instead of enforcing its contribution to zero via \eqref{chpt6:eq:QiGi}.
\end{remark}

With the optimal solution to the AFDIP established, the next section explores several special cases and their implications for solvability.

\section{Design for actuator and sensor faults}
\label{chpt6:sec:DESIGN FOR ACTUATOR AND SENSOR FAULTS}
This section discusses three practically relevant scenarios: 
actuator faults only,  sensor faults only, and combined actuator and sensor faults. 
For each case, the general solvability conditions of the AFDIP derived in Section~\ref{chpt6:sec:CLAspects} are translated into concrete design implications for admissible structure matrices $S$.

Throughout this section, $G_d$ is assumed absent. Disturbance effects can alternatively be included in $G_w$ and treated via the optimization problem \eqref{chpt6:eq:optim_problem}. In this case, the $S$-fault isolability condition reduces to
\begin{equation}
\mathrm{rank}
\begin{bmatrix}
\hat G_d^{(i)} & G_{f_j}
\end{bmatrix}
>
\mathrm{rank}
\begin{bmatrix}
\hat G_d^{(i)}
\end{bmatrix},
\quad \forall j \text{ with } S_{ij}\neq 0,
\end{equation}
where $\hat G_d^{(i)}$ is formed from those columns of $G_f$ corresponding to $S_{ij}=0$.

\subsection{Actuator faults}
Consider the system \eqref{chpt6:eq:y_OL} subject only to actuator faults, so that $n_f = n_u$ and the fault transfer matrix satisfies $G_f^{\mathrm{act}} = G_u$. Without loss of generality, the actuator fault model can be factorized as
\begin{equation}
G_f(s)=\frac{1}{d(s)}
\begin{bmatrix}
N_{11}(s) & \ldots & N_{1n_u}(s)\\
\vdots & \ddots & \vdots\\
N_{n_y1}(s) & \ldots & N_{n_yn_u}(s)
\end{bmatrix}.
\end{equation}
where $d(s)$ is a common denominator and $N(s)$ is a polynomial numerator matrix. Since multiplication by the scalar transfer function $1/d(s)$ does not affect the normal rank, all isolability conditions depend only on the column rank properties of $N(s)$.

Strong fault isolability ($S = I_{n_f}$) requires $\mathrm{rank}\, G_f = n_f$. Because $\mathrm{rank}\, G_f = \mathrm{rank}\, N$, this condition is equivalent to $N(s)$ having full column rank. Consequently, the number of sensors must satisfy $n_y \geq n_u$, and the fault columns must be linearly independent. Hence, a system with only actuator faults is strongly isolable if and only if the actuator fault transfer matrix is full column rank and at least as many sensors as actuators are available.

If $n_y < n_u$, strong isolation is impossible. In that case, the least restrictive structure is the hollow matrix $S = J_{n_f} - I_{n_f}$, corresponding to weak isolability under the assumption that only one fault occurs at a time. The rank conditions then reduce to the pairwise tests
\[
\mathrm{rank}
\begin{bmatrix}
G_{f_i} & G_{f_j}
\end{bmatrix}
>
\mathrm{rank}
\begin{bmatrix}
G_{f_i}
\end{bmatrix}
= 1,
\quad i \neq j.
\]
Since the denominator does not influence rank, this is equivalent to requiring that the columns of $N(s)$ are pairwise linearly independent. In addition, at least two sensors are required ($n_y \geq 2$), since otherwise all columns are necessarily dependent.

\subsection{Sensor faults}
Next, consider the system \eqref{chpt6:eq:y_OL} subject only to sensor faults. In this case $n_f = n_y$ and the fault model is given by
\[
G_f^{\mathrm{sens}}(s) = I_{n_y}.
\]
Hence, each fault affects exactly one measured output channel.

Since $G_f = I_{n_y}$, the fault transfer matrix is constant and full rank, with
\[
\mathrm{rank}\, G_f = n_f = n_y.
\]
Therefore, the strong isolability condition is always satisfied. A system with only sensor faults is thus inherently strongly fault isolable. Each sensor fault produces a direction in the residual space that is linearly independent of all others.

No additional requirements arise on the number of actuators, since actuator dynamics do not influence the rank of $G_f$ in this case. Moreover, because strong isolability holds, any less restrictive structure matrix $S$ (such as the hollow structure corresponding to weak isolation) is automatically achievable.

\subsection{Actuator and sensor faults}
\label{chpt6:subsec:Actuator and Sensor faults}

Finally, consider the system \eqref{chpt6:eq:y_OL} subject to both actuator and sensor faults. In this case
$n_f = n_u + n_y$ and the fault model $
G_f = \begin{bmatrix} G_f^{\mathrm{act}}(s) & G_f^{\mathrm{sens}}(s) \end{bmatrix}
= \begin{bmatrix} G_u(s) & I_{n_y} \end{bmatrix}$.

Without loss of generality, the combined fault model can be written as
\[
G_f(s)
=
\frac{1}{d(s)}
\begin{bmatrix}
N_{11}(s) & \ldots & N_{1n_u}(s) & d(s) & \ldots & 0 \\
\vdots & \ddots & \vdots & \vdots & \ddots & \vdots \\
N_{n_y1}(s) & \ldots & N_{n_yn_u}(s) & 0 & \ldots & d(s)
\end{bmatrix}.
\]
As in the actuator-only case, multiplication by the scalar factor $1/d(s)$ does not affect the normal rank. Hence, all isolability conditions depend solely on the column independence properties of the corresponding numerator matrix.

Strong isolability would require $\mathrm{rank}\, G_f = n_f = n_u + n_y$. However, since $G_f \in \mathbb{R}^{n_y \times (n_u+n_y)}$, its normal rank satisfies $\mathrm{rank}\, G_f \leq n_y$. Because $n_f = n_u + n_y > n_y$ whenever $n_u > 0$, the strong isolability condition can never be satisfied if both actuator and sensor faults are considered simultaneously. Thus, a system in which all actuators and all sensors are allowed to fail is never strongly fault isolable.

Consider next the weakest meaningful requirement, corresponding to the hollow structure $S = J_{n_f} - I_{n_f}$, which assumes that at most one fault occurs at a time. In this case, isolability reduces to pairwise column independence of $G_f$, i.e.,
\[
\mathrm{rank}
\begin{bmatrix}
G_{f_i} & G_{f_j}
\end{bmatrix}
>
\mathrm{rank}\, G_{f_i}
\qquad \forall i \neq j.
\]
Hence, weak isolability holds if and only if all fault columns of $G_f$ (both actuator and sensor columns) are pairwise independent and $n_y \geq 2$.

This condition reveals an inherent structural limitation: at least two sensors are required to distinguish any two distinct faults. In particular, if $n_y = 1$, then all fault columns are necessarily linearly dependent and isolation is impossible, even under the single-fault assumption.


\thispagestyle{empty}

\section{Experimental results} \label{chpt6:sec:EXPERIMENTAL PROOF OF PRINCIPLE}
In this section, a fault detection and isolation filter is applied in a prototype wafer stage used in the semiconductor industry. The system is equipped with 4 sensors and 13 actuators, which are all assumed to be prone to faults, hence the analysis presented in \Cref{chpt6:subsec:Actuator and Sensor faults} is used. First, the system and its control algorithm are introduced. Following this, an accurate parametric model of the system is presented. Subsequently, a tailored FDI filter design is synthesized based on an open-loop formulation to demonstrate the effectiveness of the fault diagnosis approach. The primary objective is to detect and isolate faults in each actuator and sensor.

\subsection{Prototype experimental wafer stage and control algorithm}
\label{chpt6:subsec:Prototype experimental setup}
Consider an overview of the prototype wafer stage depicted in Figure \ref{chpt6:fig:OAT_overview}. A close-up of the force frame and the bottom of the stage is depicted in Figure \ref{chpt6:fig:OAT_forceframe} and Figure \ref{chpt6:fig:OAT_Bottom}, respectively. The stage is the only moving part in the setup and is suspended by gravity compensators to reduce the required actuator forces. For this experimental case study only the actuators and sensors in the out-of-plane direction are considered. In this direction, the position of the stage is measured by four linear encoders with nanometer resolution. The system is actuated by 13 Lorentz actuators. The 13 coils are visible in Figure \ref{chpt6:fig:OAT_forceframe} and the interface on the chuck is shown in Figure \ref{chpt6:fig:OAT_Bottom}.

%
%

\begin{figure}[H]
    \centering
   \setlength{\fboxsep}{-1pt}%
    \setlength{\fboxrule}{1pt}%
    \fbox{\includegraphics[scale=.9]{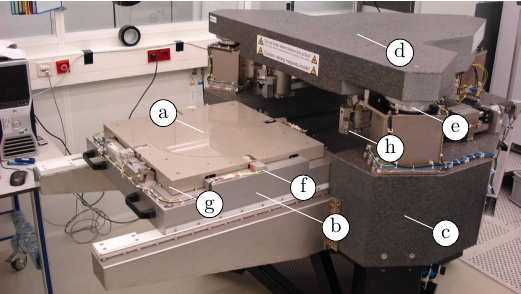}}
    \caption{Prototype experimental wafer stage setup. The moving part, the chuck, is indicated by \encircle{a} and is suspended by gravity compensators on the force frame \encircle{b}, which is on top of the base frame \encircle{c}. Currently, the chuck and force frame are slid out of the machine, whereas during operation, it is underneath the metroframe \encircle{d}. The metroframe is isolated from the fixed world through air mounts \encircle{e}. There are four Lorentz actuators in the horizontal plane of which one is indicated by \encircle{f}. These actuators apply a tangential force to the chuck. The actuators in the vertical plane are positioned between the chuck and the force frame. The position of the chuck is measured in the horizontal plane by means of capacitive sensors and measured in the vertical plane by means of linear encoders. The chuck has four scales \encircle{g} measured by the encoders \encircle{h} on the metroframe.}  
    \label{chpt6:fig:OAT_overview}
\end{figure}

\begin{figure}[H]
   \centering
   \setlength{\fboxsep}{-1pt}%
    \setlength{\fboxrule}{1pt}%
    \fbox{\includegraphics[scale=.9]{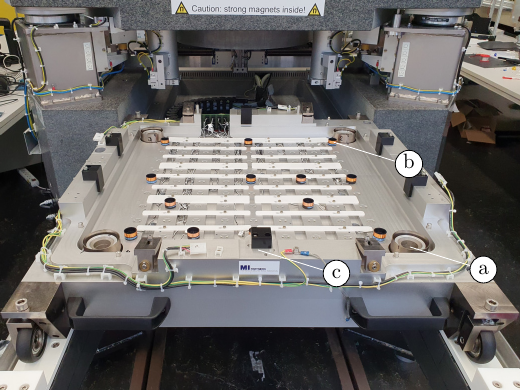}}
   \caption{Close-up of the force frame underneath the chuck. The interface of the gravity compensator to the force frame is indicated by \encircle{a}. One of the thirteen Lorentz actuators that apply a force to the chuck in vertical direction is indicated by \encircle{b}. The position of one of the in-plane actuators is indicated by \encircle{c}. }
   \label{chpt6:fig:OAT_forceframe}
\end{figure}

\thispagestyle{empty}

The system is controlled in all six degrees of freedom (DOFs) by a decentralized PID controller. To this end, the system is decoupled using input and output transformation matrices $T_{u}$ and $T_{y}$ respectively, see Figure \ref{chpt6:fig:RG_Nominal_OAT}. Each DOF is controlled by a dedicated PID controller in a diagonally structured $C$, operating at a sampling frequency of 10 kHz. The chuck follows a smooth $4^{\mathrm{th}}$-order setpoint with a stroke of $100$ $\mu$m at a frequency of $1$ Hz in the out-of-plane direction. The other two translational DOFs and three rotational DOFs are regulated to zero. 

\subsection{System identification}
\label{chpt6:subsec:Modeling and Control Algorithm}
A closed-loop multisine identification experiment is performed to obtain a best linear approximation (BLA) using the

\begin{figure}[H]
   \centering
   \setlength{\fboxsep}{-1pt}%
    \setlength{\fboxrule}{1pt}%
    \fbox{\includegraphics[scale=.9]{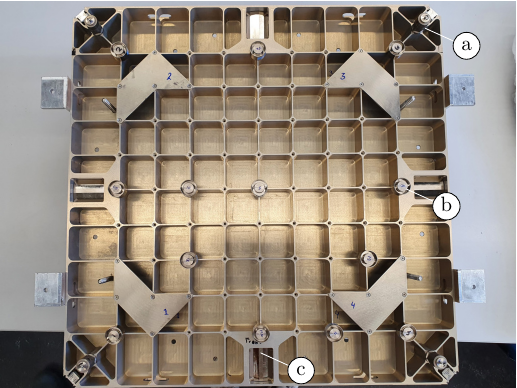}}
   \caption{Close-up of the bottom of the chuck. The interface of the gravity compensator to the chuck is indicated by \encircle{a} and the interface of the Lorentz actuator is indicated by \encircle{b}. The interface for the in-plane actuators is indicated by \encircle{c}.}
   \label{chpt6:fig:OAT_Bottom}
\end{figure}

\begin{figure}[H]
\centering
\includegraphics[width=\columnwidth]{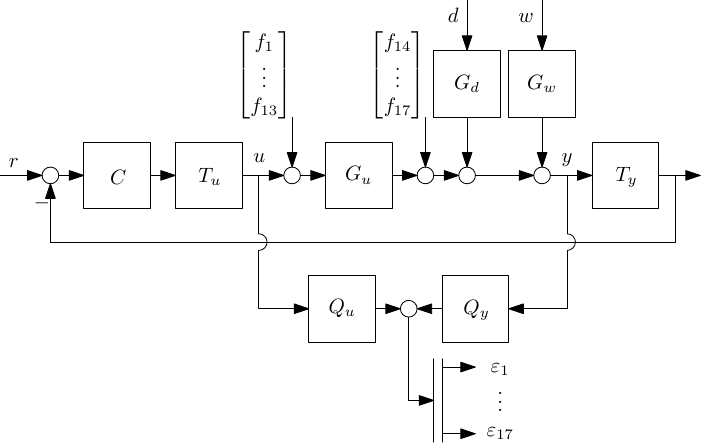}
\caption{Block diagram of the closed-loop controlled prototype wafer stage system. The 13 actuator and 4 sensor faults are indicated by $f_{1}$ to $f_{17}$ respectively.}
\label{chpt6:fig:RG_Nominal_OAT}
\end{figure}
%


\noindent robust method, see \cite{pintelonSystemIdentificationFrequency2012}. This frequency response function (FRF) is used to fit a parametric modal model of order 20. To this end, a frequency-domain modal identification algorithm is used, based on the MIMO simplified refined instrumental variable method (SRIVC) \citep{gonzalezIdentificationAdditiveContinuoustime2024,hulst2025frequency} and integrated prediction error minimization (IPEM) \citep{hulst2026structured}.

\begin{figure*}[tb]
\centering
\includegraphics[width = .9\linewidth]{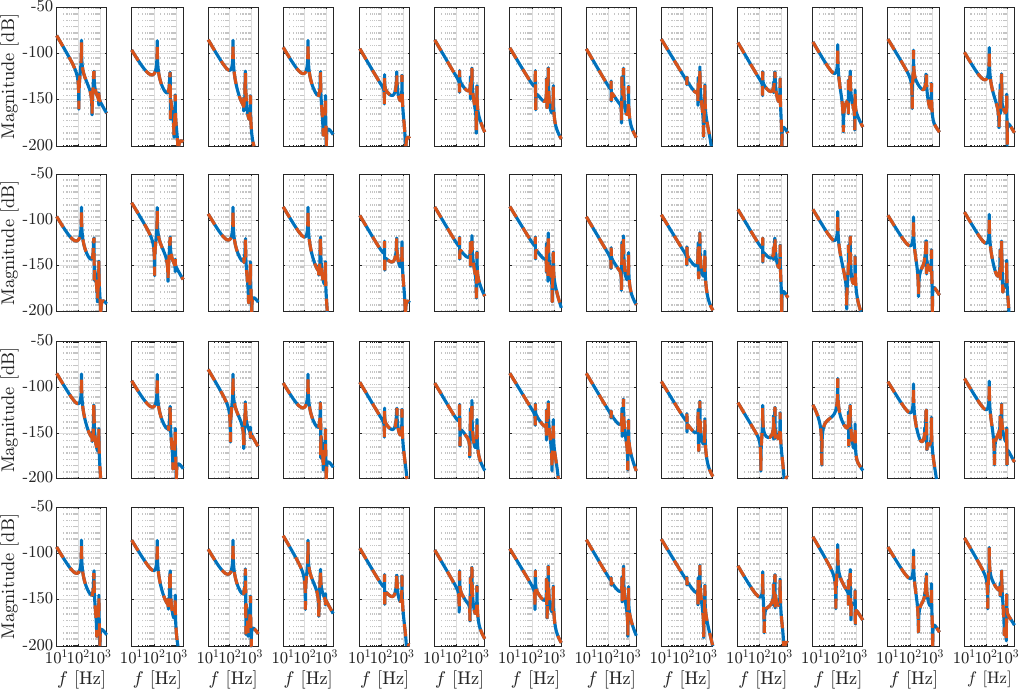}
\caption{Bode magnitude plot of the BLA in \tikzline{MatlabBlue}, measured using multisine excitation and the robust method, and the $20^{\mathrm{th}}$-order fit in \tikzdashedline{MatlabRed}, obtained using the frequency domain simplified refined instrumental variable (SRIVC) method with integrated prediction error minimization (IPEM).}
\label{chpt6:fig:CL_System3R}
\end{figure*}

The 20th-order model, denoted by $\hat{G}_{u}: \begin{bmatrix}
u_{1} & \ldots & u_{13}
\end{bmatrix}^{\top} \rightarrow \begin{bmatrix}
y_{1} & \ldots & y_{4}
\end{bmatrix}^{\top}$ is depicted in \Cref{chpt6:fig:CL_System3R} and serves as a basis for FDI filter synthesis.

\subsection{FDI Design}
The parametric model $\hat{G}_{u}$ is used to synthesize the FDI filter. The aim is to detect all possible faults, indicated in the block diagram in \Cref{chpt6:fig:RG_Nominal_OAT}. Transferring the actuator faults to the sensor side of the plant gives the fault model
\begin{equation}
G_{f} = \begin{bmatrix}
G_{f}^{\mathrm{act}} & G_{f}^{\mathrm{sens}}
\end{bmatrix},
\end{equation}
where $G_{f}^{\mathrm{act}} = \hat{G}_{u}$ and $G_{f}^{\mathrm{sens}} = I_{4}$. Since the sensors are of high quality, no disturbance and noise contributions are taken into account, so $G_{d}$ and $G_{w}$ are chosen to be void. The analysis in Section \ref{chpt6:sec:DESIGN FOR ACTUATOR AND SENSOR FAULTS} shows that strong fault isolability is not achievable. Hence, a less stringent structure $S$ is used. In this case, an almost hollow signature matrix is used, equal to
\setcounter{MaxMatrixCols}{18}
\begin{equation}
S = \resizebox{0.6\columnwidth}{!}{$\begin{bmatrix}
0 & 1 & 1 & 1 & 1 & 1 & 1 & 1 & 1 & 1 & 1 & 1 & 1 & 0 & 1 & 1 & 1 \\
1 & 0 & 1 & 1 & 1 & 1 & 1 & 1 & 1 & 1 & 1 & 1 & 1 & 1 & 0 & 1 & 1 \\
1 & 1 & 0 & 1 & 1 & 1 & 1 & 1 & 1 & 1 & 1 & 1 & 1 & 1 & 1 & 0 & 1 \\
1 & 1 & 1 & 0 & 1 & 1 & 1 & 1 & 1 & 1 & 1 & 1 & 1 & 1 & 1 & 1 & 0 \\
1 & 1 & 1 & 1 & 0 & 1 & 1 & 1 & 1 & 1 & 1 & 1 & 1 & 1 & 1 & 1 & 1 \\
1 & 1 & 1 & 1 & 1 & 0 & 1 & 1 & 1 & 1 & 1 & 1 & 1 & 1 & 1 & 1 & 1 \\
1 & 1 & 1 & 1 & 1 & 1 & 0 & 1 & 1 & 1 & 1 & 1 & 1 & 1 & 1 & 1 & 1 \\
1 & 1 & 1 & 1 & 1 & 1 & 1 & 0 & 1 & 1 & 1 & 1 & 1 & 1 & 1 & 1 & 1 \\
1 & 1 & 1 & 1 & 1 & 1 & 1 & 1 & 0 & 1 & 1 & 1 & 1 & 1 & 1 & 1 & 1 \\
1 & 1 & 1 & 1 & 1 & 1 & 1 & 1 & 1 & 0 & 1 & 1 & 1 & 1 & 1 & 1 & 1 \\
1 & 1 & 1 & 1 & 1 & 1 & 1 & 1 & 1 & 1 & 0 & 1 & 1 & 1 & 1 & 1 & 1 \\
1 & 1 & 1 & 1 & 1 & 1 & 1 & 1 & 1 & 1 & 1 & 0 & 1 & 1 & 1 & 1 & 1 \\
1 & 1 & 1 & 1 & 1 & 1 & 1 & 1 & 1 & 1 & 1 & 1 & 0 & 1 & 1 & 1 & 1 \\
1 & 1 & 1 & 1 & 1 & 1 & 1 & 1 & 1 & 1 & 1 & 1 & 1 & 0 & 1 & 1 & 1 \\
1 & 1 & 1 & 1 & 1 & 1 & 1 & 1 & 1 & 1 & 1 & 1 & 1 & 1 & 0 & 1 & 1 \\
1 & 1 & 1 & 1 & 1 & 1 & 1 & 1 & 1 & 1 & 1 & 1 & 1 & 1 & 1 & 0 & 1 \\
1 & 1 & 1 & 1 & 1 & 1 & 1 & 1 & 1 & 1 & 1 & 1 & 1 & 1 & 1 & 1 & 0  
\end{bmatrix}$}. \label{chpt6:eq:S17x17}
\end{equation}
The zeros in the top right corner of $S$ decouple the sensor faults $f_{14}$ to $f_{17}$ in the first four residuals to avoid enforcing fault-sensitivity in directions where the achievable sensitivity is too small for reliable detection. Using this structure $S$ an approximate fault detection and isolation filter is synthesized using the approach described in Section \ref{chpt6:sec:SOLUTION}.
\thispagestyle{empty}

\begin{figure*}[tb]
\centering
\includegraphics[width = .9\linewidth]{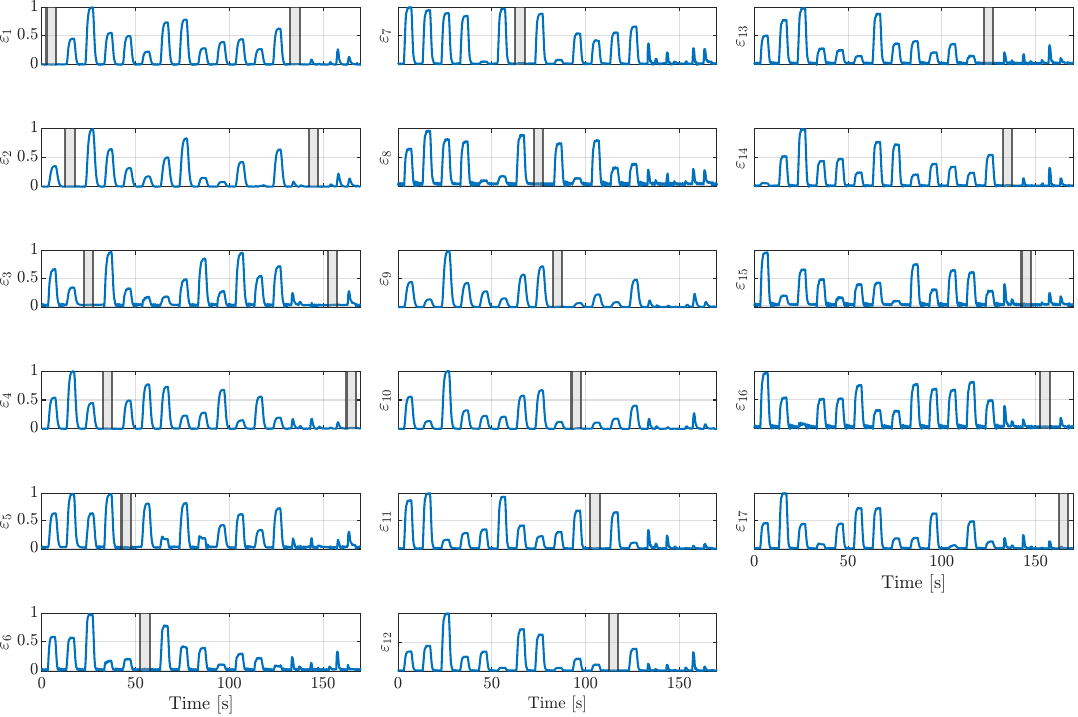}
\caption{Absolute and normalized values of the residual signals $\varepsilon_{1}$ to $\varepsilon_{17}$ \tikzline{MatlabBlue}. The faults can easily be localized using the fault signatures encoded in the structure of $S$ in \eqref{chpt6:eq:S17x17}. The corresponding time windows where the residuals are expected to be zero due to the embedded structure are indicated with the vertical lines \tikzline{black}. For instance, between $t=2.5$ s and $t=7.5$ s, all residuals fire except for $\varepsilon_{1}$, indicating a fault at actuator 1.}
\label{chpt6:fig:Residuals}
\end{figure*}

\subsection{Experimental time-domain results}
The FDI filter is synthesized and integrated into the closed-loop control system of the prototype wafer stage. The resulting filter is of order $42$, leading to a low computational load that is well suited for real-time implementation. Faults are artificially injected into the system via the channels indicated in Figure \ref{chpt6:fig:RG_Nominal_OAT}. First, each actuator fault is introduced sequentially as a step signal with magnitude $f_{1}, \ldots, f_{13} = 0.1$ N. Following this, sensor faults are injected in the form of a step signal with magnitude $10$ $\mu$m. The first fault is injected at $t = 2.5$ s until $t = 7.5$ s. Each subsquent fault is injected $5$ s after the previous fault.

 The absolute and normalized values of the residual signals are shown in \Cref{chpt6:fig:Residuals}. Clearly, the faults are easily detectable and the rootcause of the fault can be deduced by comparison of the fired signals and the fault signatures embedded in $S$. For instance, around $t=5$ s the residuals $\varepsilon_{2}$ to $\varepsilon_{17}$ fire while $\varepsilon_{1} \approx 0$ (see gray highlighted area), which indicates an additive fault at the first actuator. The residual signals can be processed, e.g., by applying suitable thresholds, to obtain a Boolean decision vector indicating the location of the active fault.
 
The residual signals exhibit small oscillations in the fault-free case, particularly with a period equal to the setpoint $r$. This undesired effect is attributed to the model-reality mismatch, resulting in a slight violation of the decoupling properties and causing a minor leakage of $r$ into the residual signals.

The actuator faults produce persistent residual signals in the case of persistent fault excitation. This implies an asymptotically non-vanishing residual signal in the case of persistent faults such as step or sinusoidal signal. The sensor faults lack strong fault detectability making detection slightly more difficult because its effect manifest in the residual only during the transient of the fault, thus the effect disappears in the residual over time, see \cite[Chapter 3]{vargaSolvingFaultDiagnosis2017} for details. However, with proper processing of the residuals, all faults are detected and isolated, showcasing the applicability of this approach to high-precision industrial equipment.

\thispagestyle{empty}

\section{Conclusion} \label{chpt6:sec:CONCLUSION}
Essential principles for fault detection and isolation for closed-loop controlled mechatronic systems are established. It is shown that the open-loop formulation may be used in closed-loop scenarios. Moreover, fundamental limitations of systems with actuator and sensor faults have been uncovered by translating generic existence requirements to these specific cases. In addition, a nullspace-based fault diagnosis system was synthesized, with experimental results demonstrating the effectiveness of the approach. These results and the experimental validation provide a solid foundation for the application of fault diagnosis to large-scale closed-loop controlled industrial systems.

   \bibliographystyle{elsarticle-harv} 
    \bibliography{references2}

@thesis{classensFaultDiagnosisUncertain2024,
  type = {Phd Thesis 1 (Research TU/e / Graduation TU/e)},
  title = {Fault {{Diagnosis}} for {{Uncertain Systems}} in {{Closed Loop}}: {{Applied}} to {{Semiconductor Equipment}}},
  author = {Classens, Koen},
  year = {2024},
  institution = {Eindhoven University of Technology},
  location = {Eindhoven, The Netherlands},
  langid = {english}
}

@Article{gaoSurveyFaultDiagnosis2015a,
  author       = {Gao, Zhiwei and Cecati, Carlo and Ding, Steven X.},
  year         = {2015},
  journal = {IEEE Trans. Ind. Electron},
  title        = {A {{survey}} of {{fault diagnosis}} and {{fault-tolerant techniques}}—{{Part I}}: {{Fault diagnosis with model-based}} and {{signal-based approaches}}},
  doi          = {10.1109/TIE.2015.2417501},
  issn         = {1557-9948},
  number       = {6},
  pages        = {3757--3767},
  volume       = {62},
  eventtitle   = {{{IEEE Transactions}} on {{Industrial Electronics}}},
  shorttitle   = {A {{Survey}} of {{Fault Diagnosis}} and {{Fault-Tolerant Techniques}}—{{Part I}}},
}

@Article{hwangSurveyFaultDetection2010,
  author       = {Hwang, I. and Kim, S. and Kim, Y. and Seah, C. E.},
  year         = {2010},
  journal = {IEEE Trans. Control Syst. Technol.},
  title        = {A {{survey}} of {{fault detection}}, {{isolation}}, and {{reconfiguration methods}}},
  doi          = {10.1109/TCST.2009.2026285},
  issn         = {1558-0865},
  number       = {3},
  pages        = {636--653},
  volume       = {18},
  eventtitle   = {{{IEEE Transactions}} on {{Control Systems Technology}}},
  ids          = {hwangSurveyFaultDetection2010a},
}

@Article{isermannModelbasedFaultdetectionDiagnosis2005,
  author       = {Isermann, Rolf},
  year         = {2005},
  journal = {Annu. Rev. Control},
  title        = {Model-Based Fault-Detection and Diagnosis - Status and Applications},
  doi          = {10.1016/j.arcontrol.2004.12.002},
  issn         = {13675788},
  number       = {1},
  pages        = {71--85},
  volume       = {29},
  langid       = {english},
}

@Article{frank1997survey,
  author  = {Frank, P. M. and Ding, X.},
  title   = {{Survey of robust residual generation and evaluation methods in observer-based fault detection systems}},
  doi     = {10.1016/S0301-4770(08)60756-3},
  issn    = {03014770},
  pages   = {403--424},
  volume  = {7},
  journal = {J. Process Control},
  year    = {1997},
}

@Article{zhangIntegratedTradeoffDesign2008,
  author       = {Zhang, Ping and Ding, Steven X.},
  journal = {Automatica},
  title        = {An Integrated Trade-off Design of Observer Based Fault Detection Systems},
  doi          = {10.1016/j.automatica.2007.11.021},
  issn         = {0005-1098},
  number       = {7},
  pages        = {1886--1894},
  volume       = {44},
  langid       = {english},
  year         = {2008},
}

@Article{wangLMIApproachIndex2007,
  author       = {Wang, J. L. and Yang, G. and Liu, J.},
  journal = {Automatica},
  title        = {An {{LMI}} Approach to  {{$H_-$}}  Index and Mixed {{$H_-$}}/{{$H_{\infty}$}} Fault Detection Observer Design},
  doi          = {10.1016/j.automatica.2007.02.019},
  issn         = {00051098},
  number       = {9},
  pages        = {1656--1665},
  volume       = {43},
  year         = {2007},
}

@InProceedings{gertlerDiagnosingParametricFaults1995,
  author     = {Gertler, J.},
  booktitle  = {1995 {{American Control Conference}} ({{ACC}})},
  title      = {Diagnosing {{parametric faults}}: {{From parameter estimation}} to {{parity relations}}},
  doi        = {10.1109/ACC.1995.529780},
  pages      = {1615-1620},
  volume     = {3},
  shorttitle = {Diagnosing Parametric Faults},
  year       = {1995},
}

@InProceedings{classens2022fault,
  author    = {Classens, K. and Mostard, M. and van de Wijdeven, J. and Heemels, W. P. M. H. and Oomen, T.},
  booktitle = {2nd Modeling, Estimation and Control Conference (MECC)},
  doi        = {10.1016/j.ifacol.2022.11.271},
  title     = {Fault Detection for Precision Mechatronics: Online Estimation of Mechanical Resonances},
  location  = {Jersey City, NJ, USA},
  pages     = {746--751},
  year      = {2022},
}

@InProceedings{hulst2025frequency,
  author    = {van der Hulst, M. and Gonz{\'a}lez, Rodrigo A and Classens, K. and Tacx, P. and Dirkx, N. and van de Wijdeven, J. and Oomen, T.},
  booktitle = {Joint 10th IFAC Symposium on Mechatronic Systems and 14th Symposium on Robotics},
  doi        = {10.1016/j.ifacol.2025.10.141},
  title     = {Frequency domain identification for multivariable motion control systems: Applied to a prototype wafer stage},
  location  = {Paris, France},
  pages     = {67-72},
  year      = {2025},
}

@Article{hulst2026structured,
  author   = {van der Hulst, M. and Gonz{\'a}lez, Rodrigo A and Classens, K. and Tacx, P. and Dirkx, N. and van de Wijdeven, J. and Oomen, T.},
  title    = {{Structured identification of multivariable modal systems}},
doi = {10.1016/j.ymssp.2026.113948},
  pages    = {113948},
  volume = {247},
 journal = {Mech. Syst. Signal Process.},
  year     = {2026},
}

@Article{Ding2000,
  author   = {Ding, S. X. and Jeinsch, T. and Frank, P. M. and Ding, E. L.},
  title    = {{A unified approach to the optimization of fault detection systems}},
doi = {10.1002/1099-1115(200011)14:7<725::AID-ACS618>3.0.CO;2-Q},
  pages    = {725--745},
  journal  = {Int. J. Adapt. Control Signal Process.},
  year     = {2000},
}

@Article{yinReviewBasicDataDriven2014,
  author       = {Yin, Shen and Ding, Steven X. and Xie, Xiaochen and Luo, Hao},
  journal = {IEEE Trans. Ind. Electron},
  title        = {A {{review}} on {{basic data-driven approaches}} for {{industrial process monitoring}}},
  doi          = {10.1109/TIE.2014.2301773},
  issn         = {1557-9948},
  number       = {11},
  pages        = {6418--6428},
  volume       = {61},
  eventtitle   = {{{IEEE Transactions}} on {{Industrial Electronics}}},
  year         = {2014},
}

@Article{leiApplicationsMachineLearning2020,
  author       = {Lei, Yaguo and Yang, Bin and Jiang, Xinwei and Jia, Feng and Li, Naipeng and Nandi, Asoke K.},
  journal = {Mech. Syst. Signal Process.},
  title        = {Applications of Machine Learning to Machine Fault Diagnosis: {{A}} Review and Roadmap},
  doi          = {10.1016/j.ymssp.2019.106587},
  issn         = {08883270},
  pages        = {106587},
  volume       = {138},
  langid       = {english},
  shorttitle   = {Applications of Machine Learning to Machine Fault Diagnosis},
  year         = {2020},
}

@Book{vargaSolvingFaultDiagnosis2017,
  author    = {Varga, Andreas},
  year      = {2017},
  title     = {{Solving fault diagnosis problems}},
  isbn      = {978-3-319-51558-8},
doi = {10.1007/978-3-031-35767-1},
publisher = {Springer Nature},
  booktitle = {Springer},
}

@Article{vargaNewComputationalParadigms2013,
  author       = {Varga, Andreas},
  year         = {2013},
  journal = {Annu. Rev. Control},
  title        = {New Computational Paradigms in Solving Fault Detection and Isolation Problems},
  doi          = {10.1016/j.arcontrol.2013.03.001},
  number       = {1},
  volume       = {37},
}

@Book{gertlerFaultDetectionDiagnosis1998,
  author   = {Gertler, J.},
  title    = {Fault {{detection}} and {{diagnosis}} in {{engineering systems}}},
  location = {{Boca Raton}},
  year     = {1998},
}

@Article{vargaComputingAchievableFault2009,
  author       = {Varga, Andras},
  journal = {IFAC Proceedings Volumes},
  title        = {On Computing Achievable Fault Signatures},
  doi          = {10.3182/20090630-4-ES-2003.00154},
  issn         = {14746670},
  number       = {8},
  pages        = {935--940},
  volume       = {42},
  langid       = {english},
  year         = {2009},
}

@Article{vargaComputingNullspaceBases2008,
  author       = {Varga, Andras},
  journal = {IFAC Proceedings Volumes},
  title        = {On Computing Nullspace Bases — a Fault Detection Perspective},
  doi          = {10.3182/20080706-5-KR-1001.01062},
  issn         = {14746670},
  number       = {2},
  pages        = {6295--6300},
  volume       = {41},
  langid       = {english},
  year         = {2008},
}

@Article{kimuraGeometricStructureObservers1977,
  author       = {Kimura, H.},
  journal      = {IEEE Trans. Autom. Control},
  title        = {Geometric Structure of Observers for Linear Feedback Control Laws},
  doi          = {10.1109/TAC.1977.1101623},
  issn         = {1558-2523},
  number       = {5},
  pages        = {846--855},
  volume       = {22},
  eventtitle   = {{{IEEE Transactions}} on {{Automatic Control}}},
  year         = {1977},
}

@InProceedings{vargaReliableAlgorithmsComputing2003,
  author    = {Varga, A.},
  booktitle = {42nd {{IEEE International Conference}} on {{Decision}} and {{Control}}},
  title     = {Reliable Algorithms for Computing Minimal Dynamic Covers},
  doi       = {10.1109/CDC.2003.1272887},
  pages     = {1873-1878},
  issn      = {0191-2216},
  year      = {2003},
}

@Book{Zhou1996a,
  author    = {Zhou, K. and Doyle, J. and Glover, K.},
  title     = {{Robust and optimal control}},
  doi       = {10.1016/s0005-1098(97)00132-5},
  publisher = {Prentice Hall},
  issn      = {01912216},
  year      = {1996},
}

@InProceedings{liuOptimalSolutionsMultiobjective2007,
  author     = {Liu, Nike and Zhou, Kemin},
  booktitle  = {2007 46th {{IEEE Conference}} on {{Decision}} and {{Control}}},
  year       = {2007},
  title      = {Optimal Solutions to Multi-Objective Robust Fault Detection Problems},
  doi        = {10.1109/CDC.2007.4434123},
  eventtitle = {2007 46th {{IEEE Conference}} on {{Decision}} and {{Control}}},
  isbn       = {1424414989},
  pages      = {981--988},
  issn       = {25762370},
  journal    = {Proceedings of the IEEE Conference on Decision and Control},
}

@Inproceedings{Glover2011,
  author  = {Glover, Keith and Varga, Andras},
  title   = {{On solving non-standard $\mathcal{H}_-/\mathcal{H}_{2/\infty}$ fault detection problems}},
  doi     = {10.1109/CDC.2011.6160723},
  booktitle  = {2011 50th {{IEEE Conference}} on {{Decision}} and {{Control}}},
  issn    = {25762370},
  number  = {2},
  pages   = {891--896},
  isbn    = {9781612848006},
  eventtitle = {2011 50th {{IEEE Conference}} on {{Decision}} and {{Control}}},
  year    = {2011},
}

@Book{pintelonSystemIdentificationFrequency2012,
  author    = {Pintelon, Rik and Schoukens, Johan},
  year      = {2012},
  title     = {System {i}dentification: {A} {f}requency {d}omain {a}pproach},
  edition   = {2nd},
  publisher = {John Wiley \& Sons},
  address   = {Hoboken, New Jersey},
}

@Article{gonzalezIdentificationAdditiveContinuoustime2024,
  author   = {Gonz{\'a}lez, Rodrigo A and Classens, Koen and Rojas, Cristian R and Welsh, James S and Oomen, Tom},
  title    = {Identification of {{additive continuous-time systems}} in {{open}} and {{closed loop}}},
volume = {173},
number ={112013},
doi = {10.1016/j.automatica.2024.112013},
  journal = {Automatica},
  year    = {2025},
}

@Article{vandenhofClosedloopIssuesSystem1998,
  author       = {{Van den Hof}, P. M. J.},
  title        = {Closed-Loop Issues in System Identification},
  doi          = {10.1016/S1367-5788(98)00016-9},
  pages        = {173--186},
  volume       = {22},
  ids          = {hofClosedloopIssuesSystem1998},
  journal      = {Annu. Rev. Control},
  publisher    = {Elsevier},
  year         = {1998},
}

@Article{Rezamand2020,
  author   = {Rezamand, M. and Kordestani, R. and Carriveau, R. and Ting, D.S.K. and Saif, M},
  title    = {{A New Hybrid Fault Detection Method for Wind Turbine Blades Using Recursive PCA and Wavelet-Based}},
doi = {10.1109/JSEN.2019.2948997},
  pages    = {2023-2033},
  volume = {20},
  number       = {4},
 journal = {IEEE Sens. J.},
  year     = {2020},
}

@Article{Neupane2025,
  author   = {Neupane, D. and Bouadjenek, M.R. and Dazeley, R. and Aryal, S.},
  title    = {{Data-driven machinery fault diagnosis: A comprehensive review
}},
doi = {10.1016/j.neucom.2025.129588},
  pages    = {129588},
  volume = {627},
 journal = {Neurocomputing},
  year     = {2025},
}

@Article{Lei2018,
  author   = {Lei, Y. and Li, N. and Guo, L. and Ningbo, L. and Tan, T. and Lin, J.},
  title    = {{Machinery health prognostics: A systematic review from data acquisition to RUL prediction
}},
doi = {10.1016/j.ymssp.2017.11.016},
  pages    = {799-834},
  volume = {104},
 journal = {Mech. Syst. Signal Process.},
  year     = {2018},
}

@Article{Mian2024,
  author   = {Mian, Z. and Deng, X. and Dongg, X. and Tian, Y. and Cao, T. and  Chen, K. and Al Jaber, T.},
  title    = {{A literature review of fault diagnosis based on ensemble learning
}},
doi = {10.1016/j.engappai.2023.107357},
  pages    = {107357},
  volume = {127},
 journal = {Eng. Appl. Artif. Intell.},
  year     = {2024},
}

@Article{Leite2025,
  author   = {Leite, D. and Andrade, E. and Rativa, D. and Maciel, A.M.A.},
  title    = {{Fault Detection and Diagnosis in Industry 4.0: A Review on Challenges and Opportunities
}},
doi = {10.3390/s25010060},
  pages    = {60},
  volume = {25},
 journal = {Sensors},
  year     = {2024},
}

@Article{Li2020,
  author   = {Li, W. and Li, H. and Gu, S. and Chen, T.},
  title    = {{Process fault diagnosis with model- and knowledge-based approaches: Advances and opportunities
}},
doi = {10.1016/j.conengprac.2020.104637},
  pages    = {104637},
  volume = {105},
 journal = {Control Eng. Pract.},
  year     = {2020},
}

@Article{Wilhelm2021,
  author   = {Wilhelm, Y. and Reimann, P. and Gauchel, W. and Mitschang, B.},
  title    = {{Overview on hybrid approaches to fault detection and diagnosis: Combining data-driven, physics-based and knowledge-based models
}},
doi = {10.1016/j.procir.2021.03.041},
  pages    = {278-283},
  volume = {99},
 journal = {Procedia CIRP},
  year     = {2021},
}

\appendix

\section{Proof of Theorem \ref{chpt5:theoremOLisCL_nygeqnu}}
\label{app:proof:theoremOLisCL_nygeqnu}
\thispagestyle{empty}

\begin{proof} 
Any fault detection filter derived from the open-loop formulation \eqref{chpt6:eq:OL_compare}, satisfies $R_{u}^{\mathrm{ol}} = 0$ and $R_{d}^{\mathrm{ol}} = 0$, or equivalently,
\begin{equation}
\begin{bmatrix}
Q_{y} & Q_{u}
\end{bmatrix} \underbrace{\begin{bmatrix}
G_{u} & G_{d} \\
I & 0
\end{bmatrix}}_{:= G^{\mathrm{ol}}} = 0. \label{eq:OL_nullspace}
\end{equation}
Let $r_{d}$ be the normal rank of $G_{d}$. Then, the normal rank of $G^{\mathrm{ol}}$ is $r^{\mathrm{ol}} = n_{u} + r_{d}$. Let $N_{l}$ be a $(n_{y} - r_{d}) \times (n_{y} + n_{u})$ basis of the left nullspace $\mathrm{null}_{l}(G^{\mathrm{ol}})$. Then all fault detection filters can be parameterized by $Q = W N_{l}$, with $W$ a suitable TFM.

Now consider the closed-loop description \eqref{chpt6:eq:compare}, which can be rewritten as
\begin{align*}
&\begin{bmatrix}
R_{r}^{\mathrm{cl}} & R_{d}^{\mathrm{cl}} & R_{w}^{\mathrm{cl}} & R_{f}^{\mathrm{cl}}
\end{bmatrix} \\ & \qquad = \begin{bmatrix}
Q_{y} & Q_{u}
\end{bmatrix} \begin{bmatrix}
G_{u} C S & S G_{d} & S G_{w} & S G_{f} \\
C S & -C S G_{d} & -C S G_{w} & -C S G_{f}
\end{bmatrix},  \\ & \qquad = \begin{bmatrix}
Q_{y} & Q_{u}
\end{bmatrix} \begin{bmatrix}
G_{u} & G_{d} & G_{w} & G_{f} \\
I & 0 & 0 & 0
\end{bmatrix} \\ & \qquad \quad \times \begin{bmatrix}
C S & -CS G_{d} & -CS G_{w} & -CS G_{f} \\
0 & I & 0 & 0 \\
0 & 0 & I & 0 \\
0 & 0 & 0 & I
\end{bmatrix}.
\end{align*}

Any fault detection filter derived from this closed-loop formulation, should satisfy $R_{r}^{\mathrm{cl}} = 0$ and $R_{d}^{\mathrm{cl}} = 0$. Thus,
\begin{align}
&\begin{bmatrix}
Q_{y} & Q_{u}
\end{bmatrix} \underbrace{\begin{bmatrix}
G_{u} C S & S G_{d} \\
C S & -C S G_{d}
\end{bmatrix}}_{:= G^{\mathrm{cl}}} = 0, \label{eq:CL_nullspace} \\
&\begin{bmatrix}
Q_{y} & Q_{u}
\end{bmatrix} \underbrace{\begin{bmatrix}
G_{u} & G_{d} \\
I & 0
\end{bmatrix}}_{= G^{\mathrm{ol}}} \underbrace{\begin{bmatrix}
CS & -CS G_{d} \\ 0 & I
\end{bmatrix}}_{:= H_{1}} = 0. \label{eq:CL_nullspace_last}
\end{align}
If $C$ is full normal rank, the normal rank of $H_{1}$ equals $n_{u} + n_{d}$, which is equal to its number of columns. As a result, $r^{\mathrm{cl}} = r^{\mathrm{ol}} = n_{u} + r_{d}$. Given that $G^{\mathrm{cl}}$ and $G^{\mathrm{ol}}$ have the same number of columns and same rank, $\mathrm{dim} \; \mathrm{null}_{l} (G^{\mathrm{cl}}) = \mathrm{dim} \; \mathrm{null}_{l} (G^{\mathrm{ol}})$.

Next, it is shown that not only the dimension of the left nullspace and rank are equal, but that $G^{\mathrm{cl}}$ and $G^{\mathrm{ol}}$ have an equal left nullspace. To this end, consider that
\begin{align*}
Q \in \mathrm{null}_{l} (G^{\mathrm{ol}}) &\implies Q G^{\mathrm{ol}} = 0\\
&\implies (Q G^{\mathrm{ol}}) H_{1} = 0 \\
&\implies Q G^{\mathrm{cl}} = 0 \\
&\implies Q \in \mathrm{null}_{l} G^{\mathrm{cl}} \\
&\implies \mathrm{null}_{l} (G^{\mathrm{ol}}) \subseteq \mathrm{null}_{l} (G^{\mathrm{cl}})
\end{align*}
From, $\mathrm{dim} \; \mathrm{null}_{l} (G^{\mathrm{cl}}) = \mathrm{dim} \; \mathrm{null}_{l} (G^{\mathrm{ol}})$ and $\mathrm{null}_{l} (G^{\mathrm{ol}}) \subseteq \mathrm{null}_{l} (G^{\mathrm{cl}})$ it can be concluded that $\mathrm{null}_{l} (G^{\mathrm{ol}}) = \mathrm{null}_{l} (G^{\mathrm{cl}})$. Hence, any filter that achieves $R_{u}^{\mathrm{ol}}=0$, $R_{d}^{\mathrm{ol}}=0$ in the open-loop formulation achieves $R_{r}^{\mathrm{cl}}=0$, $R_{d}^{\mathrm{cl}}=0$.

Subsequently, the residual dynamics are examined. First, the open-loop formulation is considered. To this end, parameterize $G_{u}$, $G_{d}$, $G_{w}$ and $G_{f}$ by a left coprime factorization (LCF)
\begin{equation*}
    \begin{bmatrix}
        G_{u} & G_{d} & G_{w} & G_{f}
    \end{bmatrix} = M^{-1} \begin{bmatrix}
        N_{u} & N_{d} & N_{w} & N_{f}
    \end{bmatrix},
\end{equation*}
where $M$ and $\begin{bmatrix}
        N_{u} & N_{d} & N_{w} & N_{f}
\end{bmatrix}$ are proper and stable factors. Consider the basis 
\begin{equation*}
    N_{l} = N_{l,d} \begin{bmatrix}
        M & - N_{u}
    \end{bmatrix},
\end{equation*}
where $N_{l,d}$ is a $(n_{y} - r_{d}) \times n_{y}$ proper stable basis of $\mathrm{null}_{l}(N_{d})$. With this basis, any fault detection filter satisfying $R_{u}^{\mathrm{ol}} = 0$ and $R_{d}^{\mathrm{ol}} = 0$ can be expressed as $Q = W N_{l,d} \begin{bmatrix}
        M & - N_{u}
    \end{bmatrix}$. Substituting this parametrization in \eqref{chpt6:eq:OL_compare} gives
\begin{align*}
&\begin{bmatrix}
R_{u}^{\mathrm{ol}} & R_{d}^{\mathrm{ol}} & R_{w}^{\mathrm{ol}} & R_{f}^{\mathrm{ol}}
\end{bmatrix} \\ & \qquad = W N_{l,d} \begin{bmatrix}
        M & - N_{u}
    \end{bmatrix} \begin{bmatrix}
G_{u} & G_{d} & G_{w} & G_{f} \\
I & 0 & 0 & 0
\end{bmatrix},\\ & \qquad = W N_{l,d} \begin{bmatrix}
0 & M G_{d} & M G_{w} & M G_{f} \\
\end{bmatrix},\\ & \qquad = W \begin{bmatrix}
0 & 0 & N_{l,d} M G_{w} & N_{l,d} M G_{f} \\
\end{bmatrix},
\end{align*}
and thus the remaining residual dynamics in the open-loop formulation equal
\begin{equation}
    \varepsilon = W N_{l,d} (N_{w} w + N_{f} f). \label{chpt6:eq:remainingresidual_OL}
\end{equation}
Next, consider the closed-loop dynamics. Since $\mathrm{null}_{l} (G^{\mathrm{ol}}) = \mathrm{null}_{l} (G^{\mathrm{cl}})$, any residual generator in closed loop satisfying $R_{r}^{\mathrm{cl}}=0$ and $R_{d}^{\mathrm{cl}}=0$ can be parameterized as $Q = W N_{l,d} \begin{bmatrix}
        M & - N_{u}
\end{bmatrix}$. Substituting this parametrization in \eqref{chpt6:eq:compare} gives

\begin{align*}
&\begin{bmatrix}
R_{r}^{\mathrm{cl}} & R_{d}^{\mathrm{cl}} & R_{w}^{\mathrm{cl}} & R_{f}^{\mathrm{cl}}
\end{bmatrix} \\ & \quad = W N_{l,d} \begin{bmatrix}
        M & - N_{u}
\end{bmatrix} \begin{bmatrix}
G_{u} C S & S G_{d} & S G_{w} & S G_{f} \\
C S & -C S G_{d} & -C S G_{w} & -C S G_{f}
\end{bmatrix},  \\ & \quad = W N_{l,d} \begin{bmatrix}
0 & M G_{d} & M G_{w} & M G_{f}
\end{bmatrix}, \\ & \quad = W \begin{bmatrix}
0 & 0 & N_{l,d} M G_{w} & N_{l,d} M G_{f} \\
\end{bmatrix}.
\end{align*}
and thus the remaining residual dynamics equal
\begin{equation}
    \varepsilon = W N_{l,d} (N_{w} w + N_{f} f),
\end{equation}
which is indeed the same as \eqref{chpt6:eq:remainingresidual_OL}, i.e. $R_{w}^{\mathrm{ol}} = R_{w}^{\mathrm{cl}}$ and $R_{f}^{\mathrm{ol}} = R_{f}^{\mathrm{cl}}$.
\end{proof}

\section{Proof of Theorem \ref{chpt5:theoremOLisCL_nystnu}} \label{app:proof:theoremOLisCL_nystnu}

\begin{proof}
    Following a similar reasoning as in the proof of Theorem \ref{chpt5:theoremOLisCL_nygeqnu}, let $N_{l}$ be a $(n_{y} - r_{d}) \times (n_{y} + n_{u})$ basis of the left nullspace $\mathrm{null}_{l}(G^{\mathrm{ol}})$. 

    Any fault detection filter derived from the closed-loop formulation, should satisfy $R_{r}^{\mathrm{cl}} = 0$ and $R_{d}^{\mathrm{cl}} = 0$, which can be written as,
\begin{align}
&\begin{bmatrix}
Q_{y} & Q_{u}
\end{bmatrix} \underbrace{\begin{bmatrix}
G_{u} & I \\
I & 0
\end{bmatrix}}_{:=H_{2}} \underbrace{\begin{bmatrix}
CS & -CS G_{d} \\ 0 & G_{d}
\end{bmatrix}}_{:=H_{3}} = 0.
\end{align}
Clearly, the normal rank of $H_{2}$ is $n_{y} + n_{u}$. If $C$ is full normal rank, the normal rank of $H_{3}$ equals $n_{y} + r_{d}$. As a result, the normal rank of $G^{\mathrm{cl}} = H_{2} H_{3}$ is $n_{y} + r_{d}$. In this case, $N_{l}$ is a $(n_{u} - r_{d}) \times (n_{y} + n_{u})$ basis of the left nullspace $\mathrm{null}_{l}(G^{\mathrm{cl}})$. Hence, $\mathrm{dim} \; \mathrm{null}_{l} (G^{\mathrm{cl}}) > \mathrm{dim} \; \mathrm{null}_{l} (G^{\mathrm{ol}})$. Following the same reasoning as Theorem \ref{chpt5:theoremOLisCL_nygeqnu}, $\mathrm{null}_{l} (G^{\mathrm{ol}}) \subseteq \mathrm{null}_{l} (G^{\mathrm{cl}})$. Hence, the filters $Q$ achieving $R_{u}^{\mathrm{ol}}=0$, $R_{d}^{\mathrm{ol}}=0$ are a subset of the filters achieving $R_{r}^{\mathrm{cl}}=0$, $R_{d}^{\mathrm{cl}}=0$.
\end{proof}
\thispagestyle{empty}

\end{document}